\documentclass[journals]{IEEEtran}
\usepackage{cite}
\usepackage{amsmath,amssymb,amsfonts}
\usepackage{algorithmic}
\usepackage{graphicx}
\usepackage{amsthm}
\usepackage{epsfig}
\usepackage{color}
\usepackage{lettrine}
\usepackage{float}
\usepackage{lipsum}
\usepackage{mathtools}
\usepackage{bbm}
\usepackage{suffix}
\usepackage[T1]{fontenc}
\usepackage[colorlinks]{hyperref} 
\usepackage[nameinlink,noabbrev]{cleveref}
\usepackage[colorinlistoftodos]{todonotes}
\usepackage{textcomp}
\def\BibTeX{{\rm B\kern-.05em{\sc i\kern-.025em b}\kern-.08em
    T\kern-.1667em\lower.7ex\hbox{E}\kern-.125emX}}

\newtheorem{mydef}{Definition}
\newtheorem{lemma}{Lemma}
\newtheorem{theorem}{Theorem}

\newcommand{\src}{$ \mathbbmss{A} $}
\newcommand{\uav}{$ \mathbbmss{U} $}
\newcommand{\des}{$ \mathbbmss{B} $}
\newcommand{\eve}{$ \mathbbmss{E} $}
\newcommand{\kkk}{$ \mathbbmss{K} $}

\newcommand{\E}{\ensuremath{\mathbb E}}

\def \treq {\stackrel{\tiny \Delta}{=}}

\DeclarePairedDelimiterX\MeijerM[3]{\lparen}{\rparen}%
{\begin{smallmatrix}#1 \\ #2\end{smallmatrix}\delimsize\vert\,#3}

\newcommand\MeijerG[8][]{%
	G^{\,#2,#3}_{#4,#5}\MeijerM[#1]{#6}{#7}{#8}}

\WithSuffix\newcommand\MeijerG*[7]{%
	G^{\,#1,#2}_{#3,#4}\MeijerM*{#5}{#6}{#7}}

\makeatletter
\def\@seccntformat#1{\@ifundefined{#1@cntformat}%
	{\csname the#1\endcsname\quad}%      default
	{\csname #1@cntformat\endcsname}%    enable individual control 
	}
\makeatother

\begin{document}

% \history{Date of publication xxxx 00, 0000, date of current version xxxx 00, 0000.}
% \doi{10.1109/ACCESS.2019.DOI}

% \title{On the Performance of Low-Altitude UAV-Enabled Secure AF Relaying with Cooperative Jamming and SWIPT}

% \author{
% \uppercase{Milad Tatar Mamaghani} and
% \uppercase{Yi Hong \IEEEmembership{Senior Member, IEEE}}
% \address{Department of Electrical and Computer Systems Engineering, Monash University, Clayton, VIC 3800, Australia }}

% % \tfootnote{``This work was supported by the Australian Research Council with DP160100528.''}
% \maketitle
% \markboth
% {}{}

\title{On the Performance of Low-Altitude UAV-Enabled Secure AF Relaying with Cooperative Jamming and SWIPT}
\author{
	
\IEEEauthorblockN{
Milad Tatar Mamaghani\IEEEauthorrefmark{1}, Yi Hong\IEEEauthorrefmark{1}, ~\IEEEmembership{Senior Member,~IEEE}}

\IEEEauthorblockA{
\IEEEauthorrefmark{1}Electrical and Computer Science Engineering Department, Monash University, Melbourne, Australia}
}
	
\date{2019-06-10}

\maketitle
\markboth{}{}

%% Set up some preparatory code -- activated fully after '\appendix'
%% (see 'The LaTeX Companion,' 2nd. ed., pp. 26f. for more details)

{\vspace{-7mm}}

% \corresp{Corresponding author: M. Tatar Mamaghani (e-mail: milad.tatarmamaghani@monash.edu).}

\begin{abstract}
This paper proposes a novel cooperative secure unmanned aerial vehicle (UAV) aided transmission protocol, where a source (Alice) sends confidential information to a destination (Bob) via an energy-constrained UAV-mounted amplify-and-forward (AF) relay in the presence of a ground eavesdropper (Eve). We adopt destination-assisted cooperative jamming (CJ) as well as simultaneous wireless information and power transfer (SWIPT) at the UAV-mounted relay to enhance physical-layer security (PLS) and transmission reliability. Assuming a low altitude UAV, we derive connection probability (CP), secrecy outage probability (SOP), instantaneous secrecy rate, and average secrecy rate (ASR) of the proposed protocol over Air-Ground (AG) channels, which are modeled as Rician fading with elevation-angel dependent parameters. By simulations, we verify our theoretical results and demonstrate significant performance improvement of our protocol, when compared to conventional transmission protocol with ground relaying and UAV-based transmission protocol without destination-assisted jamming. Finally, we evaluate the impacts of different system parameters and different UAV's locations on the proposed protocol in terms of ASR.
\end{abstract}


% \titlepgskip=-15pt

\maketitle

\section{Introduction} \label{sec:introduction}
\lettrine[lines=2]{U}{nmanned} aerial vehicle (UAV) based wireless communications has recently attracted significant research attentions, since it is envisioned to play a paramount role in establishing and/or improving ubiquitous and seamless connectivity of communication devices as well as enhancing capacity of future wireless networks \cite{Zeng2016a,Hayat2016,Cheng2019,Li2018,Liu2019}.

UAVs have been introduced as aerial relays (see \cite{Choi2014, Zeng2016, Zeng2018,Hua2018} and references therein) in support of long-distance data transmissions from source to destination in heavily shadowed environments and/or highly overloaded scenarios.  
Typically, due to UAVs mobility, UAVs require a sufficiently high energy resources, which can be supported via energy harvesting techniques such as wireless energy harvesting (WEH) \cite{Lu2015} as well as simultaneous wireless information and power transfer (SWIPT) \cite{PonnimbadugePerera2018}. WEH harvests energy in a controlled manner from the ambient radio-frequency (RF) signals, while SWIPT not only captures information signals, but also harvests energy of the same signals concurrently \cite{PonnimbadugePerera2018,Yang2018}. Specifically, a power splitting (PS) architecture is required to divide the received signal into two separate streams of different power levels, one for signal processing and the other for simultaneous energy harvesting \cite{Wu2017}.

One technical challenge of UAV-assisted communications is to guarantee physical layer security. The unique characteristics of Air-Ground (AG) channels can provide good channel condition for legitimate nodes, but, on the other hand, is prone to eavesdropping by non-legitimate nodes \cite{Zou2016,Wu2019}. Exploiting PLS techniques in UAV-assisted communications has been studied in \cite{Wang2017,Cai2019,Wang,Sun2019} (see references therein). In \cite{Wang2017}, the authors have studied PLS of a UAV-enabled mobile relaying scheme over non-fading AG channels, and showed that moving buffer-aided relay provides significant performance over static relaying in terms of secrecy rate. In \cite{Cai2019}, the authors studied the resource allocation and path-planning problem for energy-efficient secure transmission from a UAV base station to multiple users in the presence of a passive ground eavesdropper. In \cite{Wang}, employing UAV as a friendly jammer to enhance PLS of a  ground-relaying based communication has been studied. In \cite{Sun2019}, the authors have examined SWIPT-enabled secure transmission of millimeter wave (mmWave) for a UAV relay network, where the UAV feeds the energy-constrained IoT destination device in the presence of multiple ground eavesdroppers while a free-space path loss (PL) model for AG links was adopted. 

In this paper, we consider a practical scenario where a low-altitude UAV relaying is employed to assist communications between source and destination nodes. Based on the recent channel measurements in \cite{Amorim2017,Khuwaja2018}, low-altitude UAV relay channels may also suffer from small scale fading compared to high-altitude cases. Hence, we assume that AG channel models are Rician fading of different parameters. Then, we tackle the aforementioned security and energy limitation challenges and we make the following contributions.

\begin{itemize}
    \item 
   We propose a secure and energy-efficient transmission protocol, where a source sends confidential information to a destination via an energy-constrained UAV-mounted amplify-and-forward (AF) relay in the presence of a passive eavesdropper. In the protocol, we adopt the destination-assisted cooperative jamming (CJ) and SWIPT techniques at the UAV based relay for PLS improvements as well as energy harvesting. 
    \item
  We analyze the proposed secure transmission protocol in terms of reliability and security. In particular, we derive connection probability, secrecy outage probability, instantaneous secrecy rate, and average secrecy rate of the proposed protocol from source to destination via relay.  
    \item
    We conduct simulations to $i)$ verify our theoretical results, $ii)$ identify the best location of the UAV based relay that provides the best average secrecy rate, $iii)$ evaluate impacts of different system parameters on the system performance in terms of reliability and security, and finally $iv)$ validate the effectiveness and improvements of the proposed protocol, when compared to convectional transmission protocol with ground relaying and UAV-based transmission protocol without destination-assisted jamming. 
\end{itemize}

The rest of this paper is organized as below. Section II presents system model and channel model. In Section III, we propose the UAV-based transmission protocol with destination-assisted jamming and SWIPT techniques. In Section IV, we conduct performance analysis. Numerical results are given in Section V, and finally, conclusions are drawn in Section VI.

%%%%%%%%%%%%%%%%%%%%%%%	
\section{System Model and Channel Model}
%%%%%%%%%%%%%%%%%%%%%%%
In this section, we introduce system model and channel model.

%%%%%%%%%%%%%%%%%%%%%%
\subsection{System Model}
%%%%%%%%%%%%%%%%%%%%%%

\begin{figure}[t]
\center{\includegraphics[width=\columnwidth ]{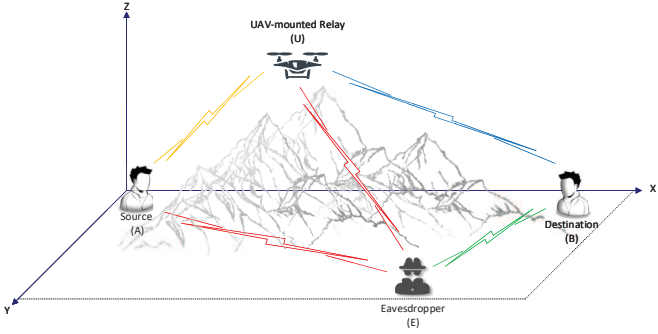}}
	\caption{\label{fig0} UAV-mounted low-altitude secure relaying communication based on destination assisted jamming and SWIPT}                                         
\end{figure}

%%%%%%%%%%%%%%%%%%%%%%%%%%
% \Figure[t]()[width=0.98\columnwidth]{SystemModel_Final.eps}
%  {UAV-mounted low-altitude secure relaying communication based on destination-assisted jamming and SWIPT.\label{fig0}}
%%%%%%%%%%%%%%%%%%%%%%%%%%

We consider a point-to-point secure transmission scheme (see Fig. \ref{fig0}), where we employ a UAV-mounted relay (\uav) to assist  confidential transmission from a source node~\src~to a legitimate destination node~\des~over heavily shadowed areas in the presence of a ground passive eavesdropper~\eve. The locations of source, destination, and eavesdropper are fixed on the ground with their 3D coordination: $W_A=(0, 0, 0)$, $W_B=(D_x, 0, 0)$, $W_E=(E_x, E_y, 0)$, while UAV is located at $W_U=(U_x, U_y, H)$ and $H$ is its altitude from ground surface.

We assume that all nodes with a single antenna operate in a half-duplex mode and the UAV relay node (\uav) adopts AF protocol. Further, we assume the relaying node (\uav) uses simultaneous wireless information and power transfer (SWIPT) technology to harvest energy from the received radio frequency (RF) signals transmitted by~\src~and~\des, while it also uses its on-board battery for maneuvering and staying stationary in the sky. 
Finally, we assume that the UAV  receiver adopts power splitting architecture (see Fig.~\ref{fig1}), where $\beta$ ($0\leq\beta\leq1$) is the power splitting ratio (PSR) identifying the portion of the harvested power from the received RF signals~and $1-\beta$ denotes the PSR for signal processing at the AF relay. 

%%%%%%%%%%%%%%%%%%%%%%%%%%%%

%%%%%%%%%%%%%%%%%%%%%%%%%%

\begin{figure}[t]
	\center{\includegraphics[width=0.7\columnwidth ]{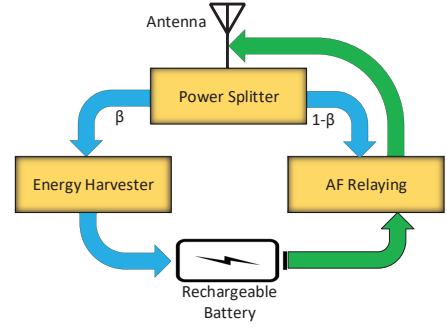}}
	\caption{\label{fig1}Power Splitting Structure for the SWIPT-enabled relaying at UAV}                                         
\end{figure}
%%%%%%%%%%%%%%%%%%%%%%%%%%%%%%%%

% \begin{figure}[t]
% \center{\includegraphics[ width=\columnwidth ]{SWIPT.eps}
% 	\caption{Power splitting structure for the SWIPT-enabled relaying at UAV.\label{fig1}}                                         
% \end{figure}
%%%%%%%%%%%%%%%%%%%%%%%%%%%%
%%%%%%%%%%%%%%%%%%%%%%
\subsection{Channel Model}
%%%%%%%%%%%%%%%%%%%%%%
Here we consider both {\em small-scale fading} and {\em large-scale PL} in setting up the channel model. 
We assume the channel between~\uav~and the ground node $\mathbbmss{G}$~$\in$ \{\src,~\des,~\eve\} is {\em Rician fading} but with different values of Rice parameters.
Exploiting channel reciprocity, the normalized channel power gain between the links~\src$\Leftrightarrow$\uav, ~\uav$\Leftrightarrow$\des,~and~\uav$\Leftrightarrow$\eve~denoted by
\begin{equation}\label{Def_S}
S_{ij}\treq|h_{ij}|^2~~~ij \in \{au, ub, ue\}
\end{equation} 
follow a square Rice distribution (i.e., non-central chi-square (nc-$\chi^2$) distribution) with two degrees of freedom, corresponding to line-of-sight (LOS) and non-line-of-sight (NLOS) components whose probability density function (PDF) and cumulative distribution function (CDF) are
\begin{align}\label{pdf_ncxhi2}
f_{{ij}}(x)&={(K_{ij}+1)} e^{-K_{ij}}\exp\Big(-{(K_{ij}+1)x}\Big)\nonumber\\
&\times\mathrm{I}_0\Big(2\sqrt{{K_{ij}(K_{ij}+1)x}}\Big),
\end{align}
and
\begin{align}
F_{ij}(x)=1-\mathrm{Q}\left(\sqrt{2K_{ij}}, \sqrt{2(1+K_{ij})x}\right),
\end{align}
where $x$ holds any non-negative value, $f_{ij}(\cdot)$ and $F_{ij}(\cdot)$ denote the PDF and the CDF for the link, 
and $\mathrm{I}_0(\cdot)$ represents the modified Bessel function of the first kind and zero-order, $K_{ij}$ (in dB) is the \emph{K-factor} given by \cite{Azari2018}
\begin{align}\label{Def_K}
K_{ij}(\theta_{ij})=\kappa_m+\left(\kappa_M-\kappa_m\right) \frac{2\theta_{ij}}{\pi},
\end{align} 
where $\kappa_M$ and $\kappa_m$ are two constants depending on the environment and transmission frequency, $\theta_{ij}$ (in radian) is the elevation angle between two given nodes. 

Furthermore, we assume the channel model between ground nodes is {\em Rayleigh fading}, {\em a special case of Rician fading with} $K=0$, and thus the channel power gain for the links \src$\Leftrightarrow$\eve, $S_{ae}\treq|h_{ae}|^2$, and \des$\Leftrightarrow$\eve, $S_{be}\treq|h_{be}|^2$ can be modeled as exponential distribution with unit scale parameters.

For large scale PL, considering the probability of LOS \cite{Zhu2018}, we adopt elevation-angle dependent PL component  
\begin{align}
\alpha_{ij}(\theta_{ij})= \frac{\alpha_L-\alpha_N}{1+\omega_1\exp\left(-\omega_2\left(\theta_{ij}-\omega_1\right)\right)}+\alpha_N,
\end{align}
where $\alpha_L$ and $\alpha_N$ represent PL exponents for LOS links between two nodes ($\theta_{ij}=\frac{\pi}{2}$), and NLOS links ($\theta_{ij}=0$), respectively, and $\omega_1$ and $\omega_2$ are environmental constants. Letting $d_{ij} \treq \| W_i -W_j \|$ be the Euclidean distance between two nodes, we let $L_{ij} (\alpha_{ij}, d_{ij})\treq d^{-\alpha_{ij}}_{ij}$ as the PL model. Note that for G2G links with $\theta_{ij}=0$, the PL model follows only NLOS components, however, for AG links the PL component is between $\alpha_L$ and $\alpha_N$ depending on the elevation-angel between nodes.

\section{Proposed Transmission Protocol}
%%%%%%%%%%%%%%%%%%%%%
\begin{figure}[t]
	\center{\includegraphics[ width=0.8\columnwidth ]{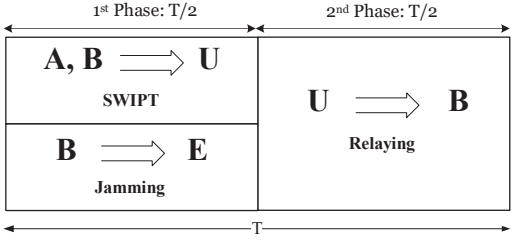} } 
	\caption{\label{fig2}Diagram of the proposed secure UAV-enabled relaying with SWIPT}  
\end{figure}
%%%%%%%%%%%%%%%%%%%%%%
% %%%%%%%%%%%%%%%%%%%%%
% \Figure[t]()[width=0.8\columnwidth]{Scenario.eps}
%  {Diagram of the proposed secure UAV relaying with SWIPT.\label{fig2}}
% %%%%%%%%%%%%%%%%%%%%%%
We consider the following two equal-duration phase secure transmission protocol with an overall duration $T$ seconds (Fig. \ref{fig2}). During the first time slot,~\src~sends information signal to~\uav~and simultaneously,~\des~transmits a {\em jamming signal} to degrade the wiretap channel of~\eve~as well as to assist the energy harvesting of UAV. Then the AF relay~\uav~receives 
\begin{align}
y_{u}&=\sqrt{(1-\beta)P_a L_{au}}h_{au} x_a +\sqrt{(1-\beta)P_b L_{bu}} h_{bu} x_b \nonumber\\
&+ \sqrt{(1-\beta)}n_u+n_p,
\end{align}
where $x_a$ and $x_b$ denote the normalized information signal from~\src, and jamming signal from~\des, i.e.,
\[ \E\{\|x_a\|^2\}=\E\{\|x_b\|^2\}=1,\] and $P_a$ and $P_b$ represent transmit power from~\src~and jamming power from~\des, which satisfy
\begin{equation}\label{Eq:toalpower}
   P_a+P_b=P,
\end{equation}
and $P$ is fixed for each frame.
Further, we assume
\begin{equation}\label{Eq:powerfraction}
P_a = \lambda P~~~P_b = (1-\lambda) P
\end{equation}
where $0<\lambda<1$ is the power allocation factor. Besides, $n_p$ represents signal processing noise at the power splitting component with power $N_p$, and $n_u \backsim {\cal N}(0,N_0)$ is the Additive white Gaussian noise (AWGN) at~\uav. The energy harvested by~\uav~from the received signals can be written as 
\begin{align}
E_{H}=\varepsilon \beta (P_a L_{au} S_{au} + P_b L_{bu} S_{bu} + N_0)T/2,
\end{align}
where $\varepsilon$ is the power conversion efficiency factor for the harvester. Here, we assume that the total harvested energy during the first phase will be used for signal transmission in the second phase and is given by
\begin{align}\label{pu}
P_u=\varepsilon \beta (P_a S_{au} +P_b S_{bu}+N_0).
\end{align}

Different from \cite{Wang2017} where no direct link was assumed to exists from~\src~to~\eve, here we consider a more general scenario during which~\eve~may also overhear the confidential messages from~\src~to~\uav~due to broadcast nature of wireless media. Specifically, if~\eve~is located on the ground that is in not so far from~\src,~\eve~can attempt to decode the received signal information based on the received signal-to-interference-plus-noise ratio (SINR)
\begin{align}\label{gammae1}
\gamma_E^{(1)}=\frac{P_a S_{ae}L_{ae}}{P_b S_{be}L_{be}+N_0}.
\end{align}
In the second phase,~\uav~forwards the scaled version of $x_u=Gy_{u}$ to~\des~with the amplification factor 
\begin{align}\label{gain}                    
G&=\sqrt{\frac{P_u}{(1-\beta)(P_aS_{au}+P_bS_{bu}+N_0)+N_p}},
\end{align}
where $P_u$ is in \eqref{pu}. The resultant signal at the node~\kkk~$\in$\{\eve, \des\} can be expressed as
\begin{align}\label{yk}
\tilde{y}_k&=\stackrel{}{\underset{\text{Information signal component}}{\underbrace{{G\sqrt{(1-\beta)P_a}h_{au}h_{uk} x_a}}}} \nonumber\\
&+\stackrel{}{\underset{\text{Destination jamming interference}}{\underbrace{{G\sqrt{(1-\beta)P_b} h_{bu}h_{uk} x_b}}}} \nonumber\\
&+ \stackrel{}{\underset{\text{Noise}}{\underbrace{{G(\sqrt{(1-\beta)}n_u+n_p)h_{uk}+n_k}}}},
\end{align}
Owning to the fact that~\des~is assumed to be able to conduct full self-interference cancellation, hence the term of jamming interference can be canceled from \eqref{yk}, whereas~\eve~acts with this part as an interference. Therefore, the received SINR at~\des~and~\eve~can be obtained as
\begin{align}\label{gammaAB}
\gamma_{A\mapsto B}&=\frac{\varepsilon\beta(1-\beta)P_aS_{au}S_{ub}L_{au}L_{ub}}{\varepsilon\beta(1-\beta+\zeta) S_{ub}L_{ub}N_0+(1-\beta)N_0+\epsilon},
\end{align}
and
\begin{align}\label{gammae2}
\gamma_E^{(2)}=\frac{\varepsilon\beta(1-\beta)P_aS_{au}S_{ue}L_{au}L_{ue}}{
	\splitfrac{\varepsilon\beta(1-\beta)P_bS_{bu}S_{ue}L_{bu}L_{ue} +(1-\beta)N_0}{
	+\varepsilon\beta\left(1-\beta+\zeta\right)S_{ue}L_{ue}N_0+\epsilon}},
\end{align}
where 
\begin{equation}\label{Noise_Ratio}
    \zeta\treq\frac{N_p}{N_0}~~~~\epsilon\treq\frac{ N_pN_0}{P_a S_{au}+P_bS_{ub}}.
\end{equation}

For the rest of the paper, we assume $\epsilon=0$ for simplicity of the result. This assumption holds for moderate/high signal-to-noise ratios (SNRs). Consequently, for the wiretap link, the total SINR at~\eve, denoted by $\gamma_E$, is given by
\begin{align}\label{gammaE}
\gamma_E=\max\{\gamma_E^{(1)}, \gamma_E^{(2)}\},
\end{align} 
which are given in \eqref{gammae1} and \eqref{gammae2}, respectively.

%%%%%%%%%%%%%%%%%%%%%%%
\section{Performance Analysis}
%%%%%%%%%%%%%%%%%%%%%%
In this section, we derive connection probability, secrecy outage probability, instantaneous secrecy rate, as well as achievable average secrecy rate of the proposed transmission protocol. In the derivations, we assume that the channel coefficients between nodes remain constants during each frame and vary from one
frame to next independently.

%%%%%%%%%%%%%%%%%%%
\subsection{Connection Probability}
%%%%%%%%%%%%%%%%%%%
\begin{mydef}
The connection probability (CP), i.e., the probability that~\des~is able to decode the transmitted signal from~\src~and correctly extract the secure information messages \cite{Mamaghani2019}, is defined as 
\begin{align}\label{pc_def}
P_c\treq\Pr\{C_M>R_t\},
\end{align}
where 
\begin{equation}\label{Def:CapacityMainLink}
C_{M}=\frac{1}{2}\log_2(1+\gamma_{A\mapsto B}),
\end{equation}
and $R_t$ denote the instantaneous capacity and transmission rate of~\src-\des~via~\uav, which is normalized by bandwidth as \cite{Laneman2004}, where $\gamma_{A\mapsto B}$ is given by (\ref{gammaAB}). 
\end{mydef}

The following theorem provides an analytical closed-form expression of $P_c$. %, providing a practical insight into the reliability performance of the UAV-based secure transmission system.
\begin{theorem}\label{pc}
We derive the CP of the secure UAV-based relaying in \eqref{pc_close}, where $\delta_t\treq2^{2R_t}-1$, $D$ and $R$ are two positive integers controlling the accuracy of \eqref{pc_close}, $\Gamma(\cdot)$ represents the gamma function, $\zeta$ is given in (\ref{Noise_Ratio}), $d,u,s,r$ are dummy variables, $K_{ij}$, $ij\in\{au,ub\}$ is given in (\ref{Def_K}), and $\mathrm{I}_{\nu}(\cdot)$ is the modified Bessel function of the first kind with the order of $\nu$.
\begin{figure*}[t]
	\begin{align}\label{pc_close}
P_c&=2(1+K_{ub})e^{-K_{au}-K_{ub}}\exp\left(-\frac{(1+K_{au})(1-\beta+\zeta)N_0\delta_t}{(1-\beta)P_a L_{au}}\right)\nonumber\\
&\times\sum_{d=0}^{D} \sum_{u=0}^{d} \sum_{s=0}^{u} \sum_{r=0}^{R} \frac{\Gamma (D+d) D^{1-2 d}  \Gamma (R+r) R^{1-2r} }{\Gamma (D-d+1)\Gamma(d+1)\Gamma(u-s+1)\Gamma(s+1)}\nonumber\\
&\times K_{au}^d (K_{au}+1)^u K_{ub}^r (1+K_{ub})^{\frac{r+u-s-1}{2}}
 \left({\frac{(1-\beta+\zeta)N_0\delta_t}{(1-\beta)P_a L_{au}}}\right)^s \left({\frac{N_0 \delta_t}{\varepsilon \beta P_a L_{au} L_{ub}}}\right)^{u-s} \nonumber\\
 &\times\mathrm{I}_{r+s-u+1}\left(2\sqrt{{\frac{(1+K_{au})(1+K_{ub})N_0 \delta_t}{\varepsilon \beta P_a L_{au} L_{ub}}}}\right),
\end{align}
\noindent\rule{\textwidth}{.5pt}%\vskip3pt
\end{figure*}

\begin{proof}
See Appendix \ref{Appendix A}.
\end{proof}
\end{theorem}
{\it Remark:} From Theorem \ref{pc}, we observe that \eqref{pc_close} is a decreasing function with respect to (w.r.t) $P_a$, implying that as the source transmission power increases, the reliability of communication improves. 

%%%%%%%%%%%%%%%%%%%%%
\subsection{Secrecy Outage Probability}
%%%%%%%%%%%%%%%%%%%%%
Following \cite{Yao2019}, when the instantaneous capacity of the wiretap link $C_E$, defined as
\begin{equation}\label{Def:CapacityWiretapLink}
C_E= \frac{1}{2}\log_2(1+\gamma_{E})
\end{equation}
is larger than the rate difference $R_e = R_t - R_s$, where $R_s$ is rate of the confidential information from~\src-\des~via~\uav, then the secrecy outage occurs and the eavesdropper is able to intercept the transmitted confidential information via~\uav.  The analytical expression for $P_{so}$ is given by Theorem below.
\begin{theorem}\label{po}
The analytical expression for secrecy outage probability $P_{so}$ is given below.
\begin{align}\label{ip}
P_{so}&=\Pr\{C_E>R_e\}=\Pr\left\{\gamma_E>\delta_e\right\}\nonumber\\
&=\Pr\{\max\left(\gamma_E^{(1)}, \gamma_E^{(2)}\right)>\delta_e\}\nonumber\\
&=1-\stackrel{}{\underset{\mathcal{L}_1}{\underbrace{\Pr\{\gamma^{(1)}_{E}\leq\delta_e\}}}}\stackrel{}{\underset{\mathcal{L}_2}{\underbrace{\Pr\{\gamma^{(2)}_{E}\leq\delta_e\}}}},
\end{align}
where $\delta_e \treq 2^{2R_e}-1$, and
\begin{align}\label{l1}
\mathcal{L}_1=1-\frac{P_aL_{ae}\exp\left(-\frac{N_0\delta_e}{P_a L_{ae}}\right)}{P_b L_{be}\delta_e + P_a L_{ae}},
\end{align}
and $\mathcal{L}_2$ is given in \eqref{l2}, where $\mathrm{K}_\nu(\cdot)$ denotes the modified Bessel function with the second kind and $\nu$-th order and $d,u,r,q,s$ are dummy variables.
\begin{figure*}[t]
\begin{align}\label{l2}
\mathcal{L}_2&=1-\frac{2(1+K_{ue})e^{-(K_{au}+K_{ub}+K_{ue})}}{
	\sqrt{K_{ub}\left(1+{\delta_e\frac{(1+K_{au})P_bL_{ub}}{(1+K_{ub})P_aL_{au}}}\right)}}\exp\left(\frac{K_{ub}/2}{\left(1+{\delta_e\frac{(1+K_{au})P_bL_{ub}}{(1+K_{ub})P_aL_{au}}}\right)}-\frac{(1+K_{ue})(1-\beta+\zeta)\delta_eN_0}{(1-\beta)P_aL_{au}}\right)\nonumber\\
&\times \sum_{d=0}^{D}\sum_{u=0}^{d}\sum_{r=0}^{u} \sum_{q=0}^{Q}\sum_{s=0}^{u-r}\Theta({D,Q,d,q,u,r,s})K_{ue}^{q}(1+K_{ue})^{u+q-s} K_{au}^{d} \mathrm{M}_{-(r+\frac{1}{2}),0}\left(K_{ub}+(1+K_{au})\delta_{e}P_bL_{ub}\right)\nonumber\\
& \times \left(\frac{1}{1+\left({\delta_e\frac{(1+K_{au})P_bL_{ub}}{(1+K_{ub})P_aL_{au}}}\right)^{-1}} \right)^{r}  \left(\frac{\delta_eN_0}{\varepsilon \beta P_a L_{au} L_{ue}}\right)^{u-r-s}\left(\frac{(1-\beta+\zeta)\delta_e N_0}{(1-\beta)P_aL_{au}}\right)^s \left({\frac {\delta_{e}\,N_{0}}{\varepsilon\,\beta\,P_{a}\,L_{{ au}}\,L_{{ ue}}}}
\right)^{\frac{q+s+1}{2}}
\nonumber\\
&\times \mathrm{K}_{q+s+1}\left(\sqrt{{\frac {\delta_{e}\,N_{0}}{\varepsilon\,\beta\,P_{a}\,L_{{ au}}\,L_
			{{ ue}}}}}\right),
\end{align}
where
\begin{align}
\Theta({D,Q,d,q,u,r,s})&={\frac {{Q}^{1-2\,q}{D}^{1-2\,d}\Gamma \left( s+1 \right) \Gamma
		\left( Q+q \right) \Gamma \left( D+d \right) }{\Gamma \left( s-u+r+1
		\right)  \left( \Gamma \left( u-r+1 \right)  \right) ^{2}\Gamma
		\left( Q-q+1 \right) \Gamma \left( D-d+1 \right)  \left( \Gamma
		\left( q+1 \right)  \right) ^{2}\Gamma \left( d+1 \right) }},
\end{align}
\noindent\rule{\textwidth}{.5pt}%\vskip3pt
\end{figure*}

\begin{proof}
See Appendix \ref{Appendix B}.
\end{proof}

\end{theorem}

%%%%%%%%%%%%%%%%%%%%%%
\subsection{Instantaneous Secrecy Rate}
%%%%%%%%%%%%%%%%%%%%%%

\begin{mydef}
The maximum achievable instantaneous secrecy rate (ISR), denoted by $C_{S}$, of the proposed UAV-enabled relaying network is defined as
\begin{align}\label{csec}
C_{S} &\treq [C_M-C_E]^+ \nonumber\\
&=\left[\frac{1}{2}\log_2\left(\frac{1+\gamma_{A\mapsto B}}{1+\gamma_E}\right)\right]^+\nonumber\\&=\left[\frac{1}{2}\log_2\left(1+\frac{\gamma_{A\mapsto B}-\gamma_E}{1+\gamma_E}\right)\right]^+,
\end{align}	
where $[x]^+\treq\max(x,0)$, $C_M$
is the capacity between~\src-\des~given in (\ref{Def:CapacityMainLink}), and $C_E$ is the capacity of the wiretap link given in (\ref{Def:CapacityWiretapLink}), and $\gamma_{A\mapsto B}$ and $\gamma_{E}$ are given in \eqref{gammaAB} and \eqref{gammaE}, respectively. 
\end{mydef}

Note that we assume that the channel coefficients between nodes remain constants during each frame and vary from one
frame to next independently. Then all parameters in (\ref{csec}) are assumed to be known except the power allocation factor $\lambda$, thereby we form the following optimization problem
\begin{equation*}
\begin{aligned}
& \underset{\lambda}{\text{maximize}}
& & C_{S}({\lambda}) \\
& \text{subject to}
& & P_a+P_b=P \\
&&& P_a=\lambda P\\
&&& P_b=(1-\lambda)P.
\end{aligned}
\end{equation*}
Considering $C_S\geq 0$ is guaranteed under optimal power allocation, the above optimization problem is equivalent to finding the optimal power allocation factor, i.e.,
\begin{align} \label{phiLambda}
\lambda^\star&=\mathrm{\arg \max}~~\phi(\lambda ) \quad \textrm{s.t.} ~~0 \leq \lambda \leq 1,
\end{align}
where 
\begin{equation}\label{Eq:SNR_Gap_Ratio}
\phi(\lambda) = \frac{\gamma_{A\mapsto B}-\gamma_E}{1+\gamma_E},
\end{equation}
which is due to the fact that $\log_2(1+\phi(\lambda))$ in (\ref{csec}) is a strictly increasing function w.r.t $\phi(\lambda)$, and thus (\ref{phiLambda}) can be solved analytically in the high SNR regime as given in Theorem \ref{theorem3}.
%%%%%%%%%%
\begin{theorem}\label{theorem3}
In large SNR regime, the function $\phi(\lambda)$ in (\ref{Eq:SNR_Gap_Ratio}) is proven to be quasi-concave w.r.t $\lambda$ in the feasible set where $0<\lambda<1$ and $\lambda^\star$ can be obtained as
\begin{align}
 \lambda^\star&=\frac{1}{1+\sqrt{\nu}},& \nu \geq 1
\end{align}
where $\nu = {\frac{X}{Y}+\frac{V}{W}} $.
\end{theorem}
\begin{proof}
See Appendix \ref{Appendix C}.
\end{proof}

%%%%%%%%%%%%%%%%%%%%%%
\subsection{Achievable Average Secrecy Rate}
%%%%%%%%%%%%%%%%%%%%%%
\begin{mydef}
The achievable average secrecy rate (ASR) is defined as
\begin{equation}\label{Def:AvgSecRate}
    \bar{C}_{S}\treq\E\{C_{S}\},
\end{equation}
Averaging over all realizations of the channels yields $\bar{C}_{S}$ in (\ref{asr_exact}).
\end{mydef}
%. as
\begin{figure*}
\begin{align} \label{asr_exact}
\bar{C}_{S}&=\E\biggl\{\Big[\frac{1}{2}\log_2(1+\gamma_{A\mapsto B})-\frac{1}{2}\log_2(1+\gamma_E)\Big]^+\biggl\}\nonumber\\
&=\frac{1}{2 \ln 2}{\int_{x=0}^{\infty}\int_{y=0}^{\infty}\int_{z=0}^{\infty}\int_{v=0}^{\infty}\int_{w=0}^{\infty}
}\left
[\ln\left(\frac{1+ \gamma_{A\mapsto B}}{1+\gamma_E}\right)\right]^+
f_{X, Y, Z, V, W}(x,y,z,v,w) dx dy dz dv dw,
\end{align} 
\noindent\rule{\textwidth}{.5pt}%\vskip3pt
\end{figure*}
Owning to the fact that the exact computation of (\ref{asr_exact}) is an arduous task and hence, instead, we derive a tight lower-bound for ASR in the Theorem below. First, we present the following worthwhile lemma, by which we then delve into the derivation of the closed-form lower-bound for the ASR in Theorem \ref{asr_lb}. 

\begin{lemma}\label{lem1}
Let $X$ be a non-central chi-square random variable with two degrees of freedom and the non-centrality parameter $\lambda$, also $b$ holds non-negative values,  then the expectation of the new random variable  $Y=\ln (X + b)$, can be obtained as follows.
\begin{align}\label{lnxpb}
\E\{\ln (X + b)\}&= \int_{0}^{\infty}\frac{1}{2}\ln(x+b)e^{-\frac{x+\lambda}{2}}{I}_0(\sqrt{\lambda x}){d}x\nonumber\\
&\treq 
\begin{cases}
        g_1(\lambda),~~~~~~~~~~\text{for } b = 0\\
        g_2(\lambda, b),~~~~~~~ \text{for } b > 0
\end{cases}
\end{align}
where $g_1(\cdot)$ and $g_2(\cdot, \cdot)$ are defined respectively as
\begin{eqnarray}\label{g1g2}
g_1(x)&\stackrel{(a)}{=}& \exp(-\frac{x}{2}) \sum_{r=0}^{R}L_r x^r , \nonumber\\
g_2(x,b)&\stackrel{(b)}{=}&\exp\left(-\frac{x}{2}\right)\sum_{r=0}^{R}C_r(b)x^r,
\end{eqnarray}
where  $R$ is some positive integer and the coefficients $L_r$ and $C_r(b)$ are given, as finite series, respectively, by
\begin{align}
L_r=\frac{\Gamma(R+r)R^{1-2r}(\Psi(r+1)+\ln 2)}{\Gamma(r+1)\Gamma(R-r+1)2^r},
\end{align}
and
\begin{align}\label{Eq:C_r(b)}
C_r(b)=\frac{\Gamma(R+r)R^{1-2r}\Phi(r,b)}{\Gamma(r+1)^2\Gamma(R-r+1)4^r},
\end{align}
wherein $\Gamma(\cdot)$ and $\Psi(\cdot)$ are the Gamma function and the Psi function respectively defined as $\Gamma(x) \treq \int_{0}^{\infty}t^{x-1}e^{-t}dt$ and $\Psi(x)\treq\frac{d}{dx}\Gamma(x)$ \cite{Gradshteyn2014}. Moreover, the function $\Phi(\cdot, \cdot)$ in \eqref{Eq:C_r(b)} is given by (see \eqref{Def:Phi})
\begin{figure*}[t]
\begin{align}\label{Def:Phi}
\Phi(i,b)&= \exp\left(\frac{b}{2}\right)\sum_{j=0}^{i}{i \choose j} (-b)^{i-j}2^{j}\bigg[\MeijerG*{3}{0}{2}{3}{1, 1}{0, 0, j+1}{\frac{b}{2}}+\ln(b)~\Gamma(j+1, \frac{b}{2})\bigg]\nonumber\\
&\stackrel{(a)}{=}\sum_{j=0}^{i}\sum_{k=0}^{j}(-1)^i{i \choose j} \left(\frac{2}{b}\right)^{j} \bigg[e^{b/2}\MeijerG*{3}{0}{2}{3}{1, 1}{0, 0, j+1}{\frac{b}{2}}+\ln(b) {j \choose k}(j-k)!\left({\frac{b}{2}}\right)^k\bigg],
\end{align}
\noindent\rule{\textwidth}{.5pt}%\vskip3pt
\end{figure*}
wherein  $\MeijerG*{a}{b}{p}{q}{\mathbf{a}, \mathbf{b}}{\mathbf{p}, \mathbf{q}}{x}$ is the analytical MeijerG function and $\Gamma(x, a) \treq \int^{\infty}_{x} e^{-t} t^{a-1} dt$ is the upper incomplete Gamma function \cite{Gradshteyn2014}. It should be noted that $(a)$ and $(b)$ are respectively obtained by applying {\cite[Eq. (4.352.1)] {Gradshteyn2014}} and {\cite[Eq. (8.352.2)] {Gradshteyn2014}} to calculate the integral expression given in \eqref{lnxpb}, and after tedious manipulations.

\end{lemma}

%%%%%%%%%%%%%
\begin{theorem}\label{asr_lb}
%%%%%%%%%%%%%
The lower bound of the average secrecy rate of the proposed secure UAV-enabled relaying system is given by
\begin{align}\label{lb_asr_expr}
\bar{C}_{LB}= \frac{1}{2\ln 2}\bigg[\ln(1+\exp(T_1))-\ln(1+T_2)\bigg]^+,
\end{align}
with
\begin{align}\label{Eq:T1}
T_1&\treq\ln\left(\frac{(1-\beta)P_aL_{au}}{(1-\beta+\zeta)N_0}\right)+g_1({\lambda_{au}})+g_1({\lambda_{ub}})\nonumber\\
&-g_2\left({\lambda_{ub}}, \frac{1-\beta}{\varepsilon \beta (1-\beta+\zeta) L_{ub}}\right),
\end{align}

\begin{align}
T_2&\treq \frac{\varepsilon\beta P_a(\lambda_{au}+2)}{\varepsilon\beta P_{b}(\lambda_{bu}+2)+\varepsilon\beta(1+\frac{\zeta}{1-\beta})N_0+\frac{N_0}{(\lambda_{ue}+2)}}\nonumber\\
&+\frac{P_a}{P_b}\frac{L_{ae}}{L_{be}}\exp\left(\frac{N_0}{P_bL_{be}}\right) \mathrm{E_1}\left(\frac{N_0}{P_bL_{be}}\right).
\end{align}

\end{theorem}
\begin{proof}
See Appendix \ref{Appendix D}.
\end{proof}

\section{Numerical Results and Discussions}

In this section, we present simulation results of connection probability, secrecy outage, and average secrecy rate in order to validate our theoretical results in the paper. We also demonstrate both reliability and security performance enhancements offered by UAV-based relaying (UR) and destination-assisted CJ transmission protocol, compared to the case using ground relaying (GR), and the case using UR but without destination-jamming.

Unless otherwise stated, we consider the following system parameters in all simulations. We assume $D_x$=10 (the normalized distance of~\src-\des~w.r.t $100$m), $H$=1.5 (the normalized height at which UAV operates w.r.t $100$m). We assume the nodes locations: $W_A=(0,0,0)$, $W_B=(D_x,0,0)$, $W_E=(\frac{4D_x}{5},1,0)$, $W_U=(\frac{D_x}{5},0,H)$. The path-loss exponents are $\alpha_L$=2 and $\alpha_N$=3.5, respectively. Besides, the network transmission rate $R_t$=0.5 bits/s/Hz, the secrecy rate $R_s$=0.2 bits/s/Hz are adopted. Further, we assume $N_0=10^{-2}$, $\zeta=2$, energy harvesting efficiency factor $\varepsilon=0.7$, $\omega_1=0.28$, $\omega_2=9.61$, $\kappa_m=1$, and $\kappa_M=10$ \cite{Zhu2018, Khuwaja2018, Mamaghani2017}. Also, the EPSA represents equal power allocation (i.e., $\lambda=0.5$) and equal power splitting ratio (i.e., $\beta=0.5$). The Monte-Carlo simulation are obtained through averaging over $100,000$ realizations of the channel coefficients.

%%%%%%%%%%%%%%%%%%%%%%
\begin{figure}[t]
	\center{\includegraphics[ width=0.8\columnwidth ]{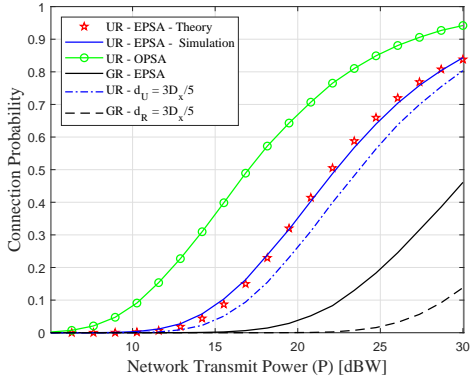}}
	\caption{\label{sim1} Connection Probability vs. Network Transmit Power}                           
\end{figure}
%%%%%%%%%%%%%%%%%%%%%
% %%%%%%%%%%%%%%%%%%%%%%
% \Figure[t]()[width= 85mm]{CP_TransmitPower.eps}
%  {Connection Probability vs. Network Transmit Power.\label{sim1}}
% %%%%%%%%%%%%%%%%%%%%%

Fig.~\ref{sim1} illustrates the CP in \eqref{pc} for the proposed transmission protocol with EPSA and OPSA, respectively, against the network transmit power $P$. Here, OPSA refers to the case with optimal $\lambda^\star$ and $\beta^\star$ for a given $P$, $d_U$ represents the horizontal projection distance of~\uav~to~\src, and $d_R$ is the distance from the ground relay to source. Here for a fair comparison, we assume $d_U = d_R$. We can observe that the CP of OPSA-based scheme gets approximately doubled, compared to EPSA-based one, in the practical range of transmit power, e.g., for $P=20$ dBW. The figure compares the simulated CP and theoretical CP in \eqref{pc} of the proposed protocol with EPSA and demonstrates that they are well matched. Fig.~\ref{sim1} also illustrates a significant CP improvement of the proposed protocol, when compared to the case using ground relaying, under the setting of EPSA. 

%%%%%%%%%%%%%%%%%%%%
\begin{figure}[t]
	\center{\includegraphics[ width=0.8\columnwidth ]{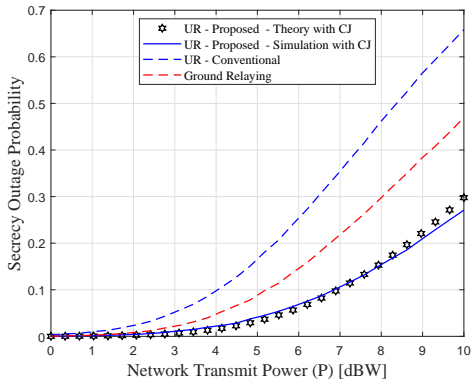}}
	\caption{\label{sim2} Secrecy Outage Probability vs. Network Transmit Power}                   
\end{figure}
%%%%%%%%%%%%%%%%%%%
% %%%%%%%%%%%%%%%%%%%%
% \Figure[t]()[width= 85mm]{SOP_TransmitPower.eps}
%  {Secrecy Outage Probability vs. Network Transmit Power.\label{sim2}}
% %%%%%%%%%%%%%%%%%%%

Fig. \ref{sim2} illustrates the simulated and theoretical SOP versus network transmit power $P$ for the proposed protocol using CJ with $\lambda = 0.7$ and demonstrates they are well matched. The figure also compares the proposed protocol, the UAV based relaying without any security technique and the one using ground relaying, and demonstrates the effectiveness of our protocol in terms of SOP thanks to the joint effect of destination CJ and SWIPT.

\begin{figure}[t]
	\center{\includegraphics[ width=0.8\columnwidth ]{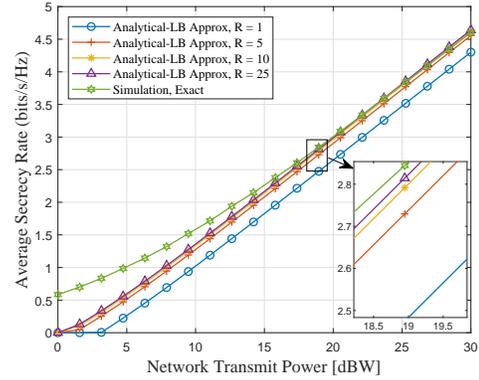}}
	\caption{\label{sim3_1} Average Secrecy Rate vs. Network Transmit Power}             
\end{figure}

% \begin{figure}[t]
% 	\center{\includegraphics[ width= 85mm ]{ESR_R_NTPower.eps}}
% 	\caption{\label{sim3_1} Average Secrecy Rate vs. Network Transmit Power}             
% \end{figure}

Fig. \ref{sim3_1} compares the lower bound of ASR (see Theorem \ref{asr_lb}) with different truncated values $R$ (see Theorem 4) and the simulated result using the exact expression \eqref{asr_exact}. We can see that they are very close,  especially when $P\geq17$dBW. Further, we observe that normalizing the gap between the exact value and the lower bound w.r.t the exact value yields the relative errors of $0.0907$, $0.0617$, and $0.0512$  for the cases of $R = 5$, $R = 10$, and, $R = 25$, respectively. This demonstrates the finite series we obtained are acceptably valid while explicitly truncated.

%%%%%%%%%%%%%%%%
\begin{figure}[t]
	\center{\includegraphics[ width=0.8\columnwidth ]{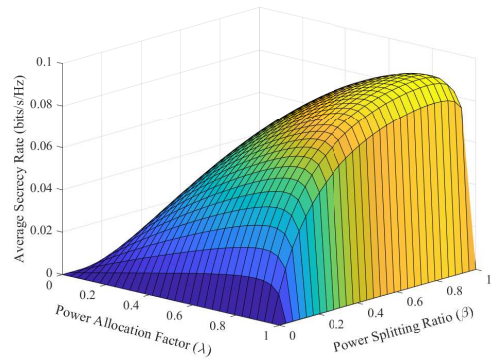}}
	\caption{\label{sim3} Average Secrecy Rate vs. Power Allocation Factor and Power Splitting Ratio}     
\end{figure}
%%%%%%%%%%%%%%%%

% %%%%%%%%%%%%%%%%
% \Figure[t]()[width= 85mm]{ASR_PSR_PAF.eps}
%  {Average Secrecy Rate vs. Power Allocation Factor and Power Splitting Ratio.\label{sim3}}
% %%%%%%%%%%%%%%%%

Fig.~\ref{sim3} depicts the impact of $\lambda$ and $\beta$ on ASR. Given a fixed power budget $P = 20$ dBW, the ASR increases when $\lambda$ and $\beta$ increase. The best ASR can be obtained at $\lambda=0.83$ and $\beta=0.8$. A higher $\lambda$ (i.e., higher $P_a$) provides higher reliability of data transmission, but when $\lambda>0.83$, ASR decreases slightly. This is due to the fact that a high source power can enhance $C_M$ and a low jamming power is sufficient to degrade eavesdropper. Overall, the plot provides a good trade-off between source power and destination jamming power in terms of ASR.

%%%%%%%%%%%%%%%%
\begin{figure}[t]
	\center{\includegraphics[ width=0.8\columnwidth ]{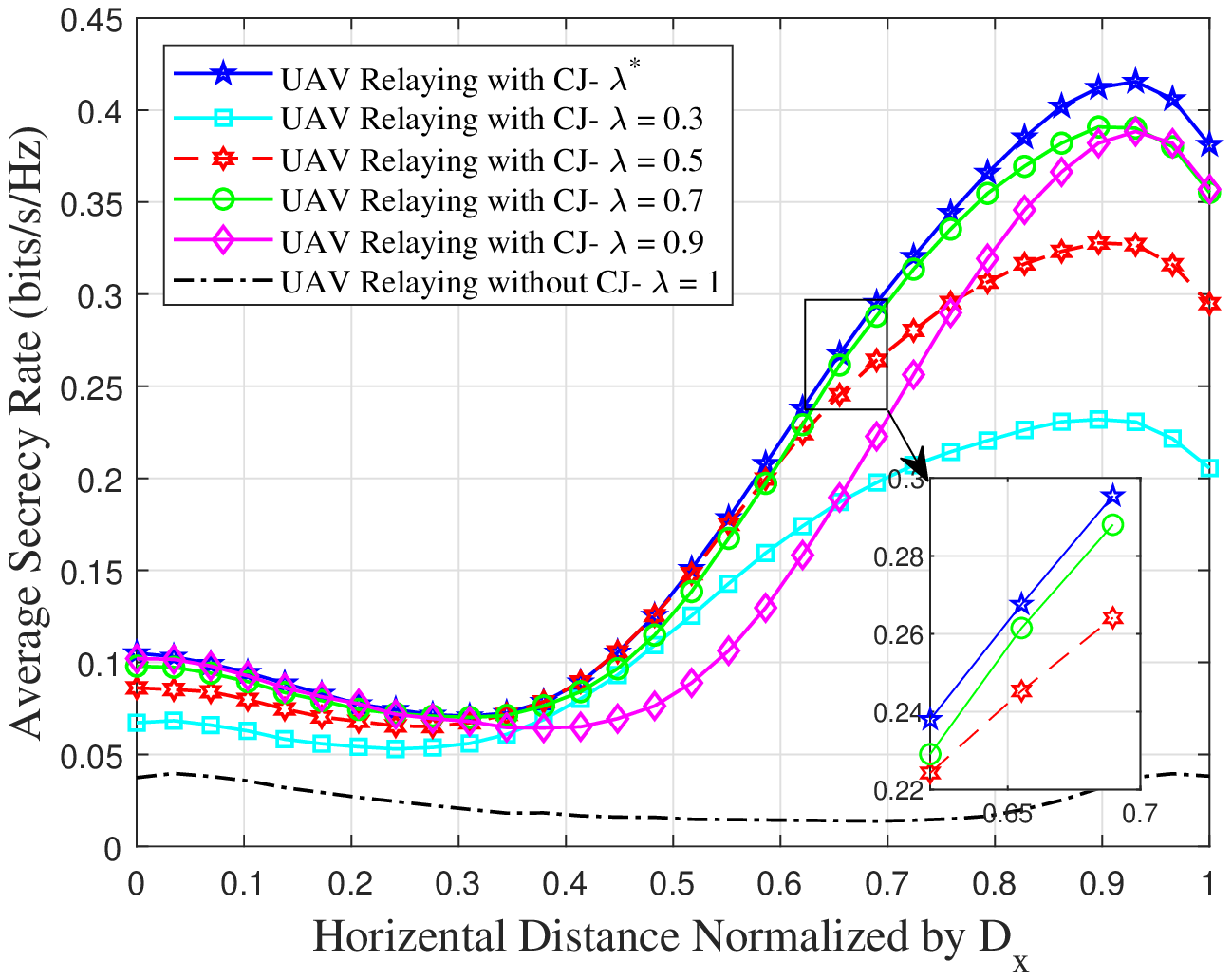}}
	\caption{\label{sim4} Average Secrecy Rate vs. Horizontal distance Factor (d)}                                         
\end{figure}
%%%%%%%%%%%%%%%%%%
% %%%%%%%%%%%%%%%%
% \Figure[t]()[width= 85mm]{ASR_DistanceFactor.eps}
%  {Average Secrecy Rate versus Horizontal Distance Ratio.\label{sim4}}
% %%%%%%%%%%%%%%%%%%
Fig. \ref{sim4} shows ASR against normalized horizontal distance of the proposed protocol with and without destination CJ for different $\lambda$ values. Here normalized horizontal distance is defined as the ratio of the horizontal distance of (\src-\uav) to that of (\src-\des). This scenario can be viewed as a relay moving from the initial location above~\src~with a direct path to the final location right above~\des~and we are seeking the best location of~\uav~in terms of ASR. 

From Fig. \ref{sim4}, we observe that the proposed protocol with CJ outperforms the one without CJ, and particularly, the proposed protocol with CJ and $\lambda^*$ leads to the best ASR, compared to other $\lambda$ values, over all different normalized horizontal distances. 
Also, we observe the best location of \uav~is $0.9D_x$. This is due to the fact that when \uav~is near \des~and far from \src, we need to allocate more $P_a$ to enhance $C_M$ for source transmission and less $P_b$, which is sufficient for destination CJ. When~\uav~is at the location of $0.2D_x-0.5D_x$ (i.e.,~\uav~is in the vicinity of eavesdropper), the wiretap link obviously reduces the ASR. Last but not least, when \uav~is within ($0.35D_x,0.65D_x$), more proportion of the power budget should be dedicated for jamming in order to improve ASR.

%%%%%%%%%%%%%%%%%%%%

\begin{figure}[t]
	\center{\includegraphics[ width=0.8\columnwidth ]{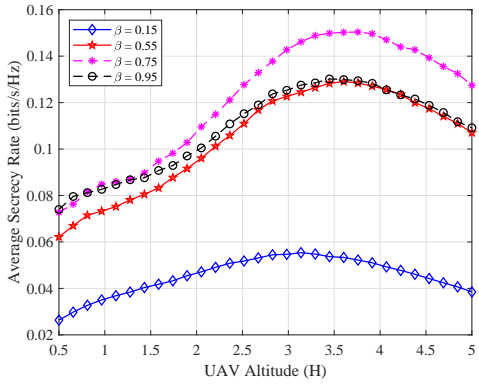}}
	\caption{\label{sim5} Average Secrecy Rate vs. Normalized UAV Altitude}                                         
\end{figure}
%%%%%%%%%%%%%%%%%%%%
% %%%%%%%%%%%%%%%%%%%%
% \Figure[t]()[width= 85mm]{ASR_Altitude_Beta.eps}
%  {Average Secrecy Rate versus Normalized UAV Altitude.\label{sim5}}
% %%%%%%%%%%%%%%%%%%%%

Fig. \ref{sim5} displays ASR against UAV altitude $H$ for the proposed protocol with different $\beta$. Firstly, we observe that $\beta=0.75$ and $H=3.5$ yield the best ASR among all. Secondly, we observe, when $H$ increases till $3.5$, ASR increases to the peak, and after that, ASR decreases. This is due to the fact that, from $H=3.5$ onwards, the attenuation factor caused by the \uav's long distance from the ground nodes significantly affects the ASR. When $0<H\leq 3.5$, the ASR is a quasi-concave function having one optimal value for given system parameters. Furthermore, for \uav~is in low altitude, e.g., $H=0.6$, increasing $\beta$ (i.e., a larger proportion of the received power is stored for EH which is then utilized in AF relaying), yields improved ASR. For higher altitudes, the ASR increases if $0< \beta < \beta^*$, or decreases if $\beta^*< \beta < 1$, where $\beta^*=0.75$. This is due to the fact that, when $\beta>\beta^*$, large UAV transmission power can lead to the improved received SNR at \des, but this also yields better signal reception at Eve.

%%%%%%%%%%%%%%%%%
\section{Conclusions}
%%%%%%%%%%%%%%%%%
In this paper, we proposed a secure and energy-efficient source-UAV-destination transmission protocol with the aid of destination cooperative jamming and SWIPT at the relay, when a passive eavesdropper exists. We derived the connection probability, secrecy outage probability, instantaneous secrecy rate, and average secrecy rate of the proposed transmission protocol from~\src~to~\des~via a stationary~\uav~and verified them via simulations. Further, by simulations, we demonstrate significant performance improvement of our protocol, when compared to conventional transmission protocol with ground relaying and UAV-based transmission protocol without destination-assisted jamming. We also identified the best location of~\uav~that provides the optimal ASR, given a fixed eavesdropper location. Finally, we evaluated the impacts of different system parameters on the system performance. For future work, we will consider the extension to flying UAV relaying with extra degrees of freedom in the system design.

\appendices
\numberwithin{equation}{section}
\makeatletter 
\newcommand{\section@cntformat}{Appendix \thesection:\ }
\makeatother

\section{Derivation of Connection Probability}\label{Appendix A}
The connection probability $P_c$ given by \eqref{pc_def} can be obtained as follows (see \eqref{prof_pc})
\begin{figure*}
\begin{align}\label{prof_pc}
P_c&=\Pr\left\{X> A + \frac{B}Y\right\}=1-\E_{Y}\left\{F_{X\mid Y}\left(A+\frac{B}{Y}\right)\right\}\nonumber\\
&\hspace{-5mm}=(1+K_{ub})e^{-K_{ub}}\int_{0}^{\infty}\mathrm{Q}\left(\sqrt{2K_{au}}, \sqrt{2(1+K_{au})\left(A+\frac{B}{y}\right)}\right) e^{-(1+K_{ub})y}\mathrm{I}_{0}\left(2\sqrt{K_{ub}(1+K_{ub})}\sqrt{y}\right)dy\nonumber\\
&\hspace{-5mm}\stackrel{(a)}{=}(1+K_{ub})e^{-K_{au}-K_{ub}-(1+K_{au})A}\sum_{d=0}^{D} \sum_{u=0}^{d} \sum_{s=0}^{u}\frac{ {\Gamma (D+d) D^{1-2 d}  \Gamma (R+r) K_{au}^d (K_{au}+1)^u A^s B^{u-s} {u\choose s}}  }{\Gamma (D-d+1)\Gamma(d+1)\Gamma(u-s+1)\Gamma(s+1)} \nonumber\\
&\times \int_{0}^{\infty}y^{-(u-s)}\exp\left(-(1+K_{ub})y-(1+K_{au})\frac{B}{y}\right)\mathrm{I}_0\left(2\sqrt\{K_{ub}(1+K_{ub})y\}\right)dy,
\end{align}
\noindent\rule{\textwidth}{.5pt}%\vskip3pt
\end{figure*}
where $X\treq S_{au}$, $Y \treq S_{ub}$, \[A\treq{\frac{(1-\beta+\zeta)N_0\delta_t}{(1-\beta)P_a L_{au}}}~~~B\treq{\frac{N_0 \delta_t}{\varepsilon \beta P_a L_{au} L_{ub}}}\] 
and $\mathrm{Q}(a,b)$ is the first-order Marqum Q-function, and $\mathrm{I}_{\nu}(\cdot)$ is the modified Bessel function of the first kind with $\nu$th order. Furthermore, $(a)$ follows from plugging the approximate expressions for $Q(\cdot, \cdot)$ and $I_{0}(\cdot)$ given respectively as \cite{Cao2016} 
\begin{align}
I_{0}(y)=\sum_{r=0}^{R}\frac{\Gamma(R+r)R^{1-2r}}{\Gamma^2(r+1)\Gamma(R-r+1)}\left(\frac{y}{2}\right)^{2r},
\end{align}
and
\begin{align}\label{Q_approx}
\mathrm Q(x,y)&=\sum_{d=0}^{D}\sum_{u=0}^{d}\frac{\Gamma(D+d)D^{1-2d}x^{2d}y^{2u} }{\Gamma(D-d+1)d!u!2^{d+u}}e^{-\frac{x^2+y^2}{2}},
\end{align}
Finally, applying the equation {\cite[Eq. (3.471.9)] {Gradshteyn2014}} for calculating the integral expression of the last equation of \eqref{prof_pc}, one can reach, after tedious manipulations, at \eqref{pc_close} and so the proof is done.

\section{Derivation of Secrecy Outage Probability}\label{Appendix B}
To obtain $P_{so}$ we need to calculate $\mathcal{L}_1$ and $\mathcal{L}_2$. Henceforth, in order to calculate $\mathcal{L}_1$, let define $S_{ae}\treq V$ and $S_{be} \treq W$ while rewriting \eqref{l1} as
\begin{align}\label{l1_app}
\mathcal{L}_1&=\Pr\left\{V<\frac{P_b L_{be}\delta_e }{P_a L_{ae}} W+ \frac{N_0 \delta_e}{P_a L_{ae}}\right\}\nonumber\\
&\hspace*{-5mm}=\E_{W}\left\{F_{V|W}\left(\frac{P_b L_{be}\delta_e }{P_a L_{ae}} w+ \frac{N_0 \delta_e}{P_a L_{ae}}\right)\right\}\nonumber\\
&\hspace*{-5mm}=1-\int_{0}^{\infty}\exp\left(-\bigg[\frac{P_b L_{be}\delta_e}{P_a L_{ae}}+1\bigg]w \hspace*{-1mm}-\hspace*{-1mm}\frac{N_0\delta_e}{P_a L_{ae}}
\hspace*{-1mm}\right)dw,
\end{align}
calculating the last integral results in \eqref{l1}. 

Now, we focus on obtaining an analytical expression for $\mathcal{L}_2$ as follows. 
By defining the auxiliary variables $a_1\treq \delta_e\frac{P_b L_{ub}}{P_a L_{au}}$, $a_2 \treq \delta_e\frac{N_0}{\varepsilon \beta P_a L_{au} L_{ue}}$, $a_3 \treq \delta_e\frac{(1-\beta+\zeta)N_0}{(1-\beta)P_aL_{au}}$, and letting $X \treq S_{au}$, $Y \treq S_{ub}$, and $Z \treq S_{ue}$,  one can rewrite $\mathcal{L}_2$ given in \eqref{l2} as
\begin{align}\label{ip_gammae2}
\mathcal{L}_2&=\Pr\{X \leq a_1 Y+a_2 Z^{-1}+a_3\}\nonumber\\
&=\E_{Z}\{\E_{Y|Z}\{F_{X|{Y,Z}}(a_1y+a_2z^{-1}+a_3)\}\nonumber\\
&=1-\int_{0}^{\infty}\Xi(z)f_{Z}(z)dz,
\end{align}
where $ \Xi(z) \treq \int_{0}^{\infty}\mathrm{Q}\left(\sqrt{a}, \sqrt{b y+c(z)}\right) f_{Y}(y)dy$ in which $a \treq 2K_{au}$, $b \treq 2(1+K_{au}) a_1$, $c(z) \treq 2(1+K_{au}) (a_2 z^{-1}+a_3)$. Then, $\Xi(z)$ is calculated in a closed-form expression  using \eqref{Q_approx} given in Appendix A as (see \eqref{xhi_closed}).
\begin{figure*}[t]
	\begin{align}\label{xhi_closed}
	\Xi(z)&=\sum_{d=0}^{D}\sum_{u=0}^{d}\sum_{r=0}^{u} \frac{(1+K_{ub})e^{-K_{ub}}\Gamma(D+d)D^{1-2d}a^{d}b^{r}{u \choose r} c^{u-r}\exp\left(-\frac{a+c}{2}\right) }{\Gamma(D-d+1)d!u!2^{d+u}}\nonumber\\
	&\hspace{20mm}\times\int_{0}^{\infty}y^r \exp\left(-\left(\frac{b}{2}+K_{ub}+1\right)y\right)\mathrm{I}_0\left(2\sqrt{K_{ub}(1+K_{ub})}\sqrt{y}\right)dy\nonumber\\
	&\stackrel{(b)}{=}\sum_{d=0}^{D}\sum_{u=0}^{d}\sum_{r=0}^{u} \frac{(1+K_{ub})e^{-K_{ub}}\Gamma(D+d)D^{1-2d}a^{d}(\frac{b}{\tilde{b}})^{r}{u \choose r}r! c^{u-r} }{\Gamma(D-d+1)d!u!2^{d+u}\tilde{c}\sqrt{\tilde{b}}}\nonumber\\
	&\hspace{20mm}\times\exp(-\frac{a+c}{2}+\frac{\tilde{c}^2}{2\tilde{b}}) \mathrm{M}_{-(r+\frac{1}{2}),0}(\frac{\tilde{c}^2}{\tilde{b}}),
	\end{align}
	\noindent\rule{\textwidth}{.5pt}\\%\vskip3pt
	\begin{align}\label{L2_closed}
	\mathcal{L}_2&=1-\sum_{d=0}^{D}\sum_{u=0}^{d}\sum_{r=0}^{u} \sum_{q=0}^{Q}\sum_{s=0}^{u-r}\frac{\Gamma(Q+q)Q^{1-2q}\Gamma(D+d)D^{1-2d}a^{d}(\frac{b}{\tilde{b}})^{r}{u \choose r}r! {{u-r} \choose s} b_3^{u-r-s} b_2^s c_1^q }{\Gamma(Q-q+1)\Gamma^2(q+1)\Gamma(D-d+1)d!u!2^{d+u-1}\tilde{c}\sqrt{\tilde{b}}}\nonumber\\
	&\hspace{20mm}\times (1+K_{ub})(1+K_{ue})e^{-(K_{ub}+K_{ue})} \exp(-\frac{a}{2}+\frac{\tilde{c}^2}{2\tilde{b}}-\frac{b_2}{2}) \left(\frac{b_3}{2b_1}\right)^{\frac{q+s+1}{2}}\nonumber\\
	&\hspace{20mm}\times\mathrm{M}_{-(r+\frac{1}{2}),0}\left(\frac{\tilde{c}^2}{\tilde{b}}\right)\mathrm{K}_{q+s+1}\left(\sqrt{\frac{b_1b_3}{2}}\right),
	\end{align}
	\noindent\rule{\textwidth}{.5pt}\\%\vskip3pt
\end{figure*}
in which $\tilde{b} \treq \frac{b}{2}+K_{ub}+1$, $\tilde{c} \treq \sqrt{{K_{ub}(1+K_{ub})}}$. Moreover, $(b)$ comes from using  {\cite[Eq. (6.643.2)] {Gradshteyn2014}} to calculate the integral term, wherein $\mathrm{M}(\cdot)$ is the Whittaker M-function. Next, in order to obtain a closed-form expression  for the single-form integral given by \eqref{ip_gammae2} we further calculate as (see \eqref{L2_closed}), where $c_1\treq K_{ue}(K_{ue}+1)$, $b_1 \treq K_{ue}+1$, $b_2 \treq 2(K_{ue}+1)a_3$, $b_3 \treq 2(K_{ue}+1) a_2$, and $\mathrm{K}_\nu(\cdot)$ denotes the modified Bessel function with the second kind and $\nu$-th order.
Note that the above equation achieved by applying {\cite[Eq. (3.478.4)] {Gradshteyn2014}}. Finally, through tedious mathematical manipulations we can achieve the closed-form expression for the SOP given by \eqref{l2}, and hence, the proof is done.

\section{Derivation of optimal $\lambda^\star$}\label{Appendix C}
Taking the auxiliary variables 
\[c_1 \treq \frac{\varepsilon \beta (1-\beta) P X Y}{\varepsilon \beta (1-\beta + \zeta) Y N_0 + (1-\beta) N_0}, c_2 \treq \frac{P V}{N_0}\]
\[c_3 \treq \frac{P W}{N_0}, c_4 \treq \frac{\varepsilon \beta (1-\beta) P X Z}{\varepsilon \beta (1-\beta + \zeta) Z N_0 + (1-\beta) N_0}\] \[c_5 \treq \frac{\varepsilon \beta (1-\beta) P Y Z}{\varepsilon \beta (1-\beta + \zeta) Z N_0 + (1-\beta) N_0}\] 
we rewrite the function $\Phi(\lambda)$, considering high Transmit SNR approximation, as
\begin{align}
\Phi(\lambda) \approx \frac{c_1c_3c_5\lambda(\lambda-1)^2}{b_2\lambda^2+b_1\lambda+b_0},
\end{align}
where $b_2 = (c_4-c_5)c_3+c_2c_5$, $b_1 = (2c_5-c_4)c_3-c_2c_5-c_2-c_4$, $b_0 = -c_3c_5$. Now, taking the first derivation of the function $\Phi$ w.r.t $\lambda$, i.e., $\Phi_1(\lambda) \treq \frac{d \Phi(\lambda)}{d \lambda}$ results a rational polynomial function with always-positive denominator and the numerator of forth-order, where it can readily be found that $0$ and $1$ are the roots of both $\Phi(\lambda)$ and its first derivative $\Phi_1(\lambda)$. Hence the remained two roots are the solutions of the second order polynomial given as
\begin{align}
(\nu-1)\lambda^2+2\lambda-1 = 0
\end{align}
where $\nu = \frac{c_2c_5+c_3c_4}{c_3c_5}$. As such, considering the feasible set of $0 < \lambda <1$, we have three cases as 
\begin{enumerate}
    \item 
    $\nu<1$, no optimum solution
    \item
     $\nu=1$,  $\lambda^\star=\frac{1}{2}$ 
    \item
    $\nu > 1$, $\lambda^\star=\frac{1}{1+\sqrt{\nu}}$
\end{enumerate}
Furthermore, note that it is easy to prove $\Phi_1(\lambda)$ holds positive values in $(0,\lambda^\star]$ and are negative in $(\lambda^\star,1)$ and the extremum determines the maximum of the $\Phi(\lambda)$ \cite{convexBoyd}, and hence the proof is complete.

%%%%%%%%%%%%%%%%%%%%%%%%%%
\section{Derivation of Lower-bounded ASR}\label{Appendix D}
%%%%%%%%%%%%%%%%%%%%%%%%%%

In order to obtain a closed-form lower-bound expression for $\E\{C_S\}$, one can write, following the Jensen's inequality, as \cite{Mamaghani2018} (see \eqref{c_lb}).
\begin{figure*}[t]
\begin{align}\label{c_lb}
\E\{C_S\} &\geq \frac{1}{2\ln 2}\bigg[\E\{\ln(1+\gamma_{A\mapsto B})\} -\E\{\ln(1+\gamma_E)\}\bigg]^+ \nonumber\\
&\geq \frac{1}{2\ln 2}\bigg[\ln\Big(1+\exp(\E\{\ln(\gamma_{A\mapsto B})\})\Big) -\ln\Big(1+\E\{\gamma^{(1)}_E\}+\E\{\gamma^{(2)}_E\}\Big)\bigg]^+,
\end{align}
\noindent\rule{\textwidth}{.5pt}%\vskip3pt
\end{figure*}
In \eqref{c_lb}, the term $\E\{\ln(\gamma_{A\mapsto B})\}$ is further calculated as
\begin{align}\label{expectAB}
\E\{\ln(\gamma_{A\mapsto B})\}&\geq\ln\left(\frac{(1-\beta)P_aL_{au}}{(1-\beta+\zeta)N_0}\right)+\E\{\ln(S_{au} S_{ub})\} \nonumber\\
&\hspace{-15mm} -\E\left\{\ln\left({S_{ub}}+ \frac{1-\beta}{\varepsilon\beta(1-\beta+\zeta)L_{ub}}
\right)\right\}\treq{T_1},
\end{align}
where $T_1$ can be analytically obtained using Lemma \ref{lem1} as given in \eqref{Eq:T1}. Letting $T_2\treq \E\{\gamma^{(1)}_E\} + \E\{\gamma^{(2)}_E\}$, the remained parts $ \E\{\gamma^{(1)}_E\}$ and $ \E\{\gamma^{(2)}_E\}$ in \eqref{c_lb} are derived as
\begin{align}
    \E\{\gamma^{(1)}_E\} &=\frac{P_a}{P_b}\frac{L_{ae}}{L_{be}}\exp\left(\frac{N_0}{P_bL_{be}}\right) \mathrm{E_1}\left(\frac{N_0}{P_bL_{be}}\right),
\end{align}
\begin{figure*}[t]
\begin{align}
    \E\{\gamma^{(2)}_E\} \stackrel{(b)}{\approx} \frac{\varepsilon\beta(1-\beta)P_a(\lambda_{au}+2)(\lambda_{ue}+2)}{\varepsilon\beta(1-\beta)P_{b}(\lambda_{bu}+2)(\lambda_{ue}+2)+\varepsilon\beta(1-\beta+\zeta)(\lambda_{ue}+2)N_0+(1-\beta)N_0},
\end{align}
\noindent\rule{\textwidth}{.5pt}%\vskip3pt
\end{figure*}
where $\mathrm{E}_1(\cdot)=\int_{1}^{\infty}e^{-tx}t^{-a}dx$  is the exponential integral, $(b)$  follows from the approximate given in \cite{Bjornson2013} along considering the tight lower bound for $E\left\{\frac{1}{X}\right\}$ given in \cite{Moser2008}. However that $E\left\{\frac{1}{X}\right\}$ can be readily obtained inasmuch as  it equals to the first derivative of $g_1(x)$, already given in \eqref{g1g2}, w.r.t $x$, we use this tight approximation for simplicity.

\bibliographystyle{IEEEtran}
\bibliography{MyPaper_References}

\begin{thebibliography}{10}
\providecommand{\url}[1]{#1}
\csname url@samestyle\endcsname
\providecommand{\newblock}{\relax}
\providecommand{\bibinfo}[2]{#2}
\providecommand{\BIBentrySTDinterwordspacing}{\spaceskip=0pt\relax}
\providecommand{\BIBentryALTinterwordstretchfactor}{4}
\providecommand{\BIBentryALTinterwordspacing}{\spaceskip=\fontdimen2\font plus
\BIBentryALTinterwordstretchfactor\fontdimen3\font minus
  \fontdimen4\font\relax}
\providecommand{\BIBforeignlanguage}[2]{{%
\expandafter\ifx\csname l@#1\endcsname\relax
\typeout{** WARNING: IEEEtran.bst: No hyphenation pattern has been}%
\typeout{** loaded for the language `#1'. Using the pattern for}%
\typeout{** the default language instead.}%
\else
\language=\csname l@#1\endcsname
\fi
#2}}
\providecommand{\BIBdecl}{\relax}
\BIBdecl

\bibitem{Zeng2016a}
Y.~Zeng, R.~Zhang, and T.~J. Lim, ``Wireless communications with unmanned
  aerial vehicles: Opportunities and challenges,'' \emph{\textit{{IEEE} Commun.
  Mag.}}, vol.~54, no.~5, pp. 36--42, 2016.

\bibitem{Hayat2016}
S.~Hayat, E.~Yanmaz, and R.~Muzaffar, ``Survey on unmanned aerial vehicle
  networks for civil applications: A communications viewpoint,''
  \emph{\textit{{IEEE} Commun. Surveys Tuts.}}, vol.~18, no.~4, pp. 2624--2661,
  2016.

\bibitem{Cheng2019}
F.~Cheng, G.~Gui, N.~Zhao, Y.~Chen, J.~Tang, and H.~Sari, ``{UAV} relaying
  assisted secure transmission with caching,'' \emph{\textit{{IEEE} Trans.
  Commun.}}, pp. 1--1, 2019.

\bibitem{Li2018}
B.~Li, Z.~Fei, and Y.~Zhang, ``{UAV} communications for {5G} and beyond: Recent
  advances and future trends,'' \emph{\textit{{IEEE} Internet Things J.}}, pp.
  1--1, 2018.

\bibitem{Liu2019}
X.~{Liu}, Z.~{Li}, N.~{Zhao}, W.~{Meng}, G.~{Gui}, Y.~{Chen}, and F.~{Adachi},
  ``Transceiver design and multi-hop {D2D} for {UAV} {IoT} coverage in
  disasters,'' \emph{\textit{{IEEE} Internet Things J.}}, pp. 1--1, 2019.

\bibitem{Choi2014}
D.~H. {Choi}, S.~H. {Kim}, and D.~K. {Sung}, ``Energy-efficient maneuvering and
  communication of a single {UAV}-based relay,'' \emph{\textit{IEEE Trans.
  Aerosp. Electron. Syst.}}, vol.~50, no.~3, pp. 2320--2327, July 2014.

\bibitem{Zeng2016}
Y.~Zeng, R.~Zhang, and T.~J. Lim, ``Throughput maximization for {UAV}-enabled
  mobile relaying systems,'' \emph{\textit{{IEEE} Trans. Commun.}}, vol.~64,
  no.~12, pp. 4983--4996, 2016.

\bibitem{Zeng2018}
S.~{Zeng}, H.~{Zhang}, K.~{Bian}, and L.~{Song}, ``{UAV} relaying: Power
  allocation and trajectory optimization using decode-and-forward protocol,''
  in \emph{\textit{Proc. IEEE ICC}}, May 2018, pp. 1--6.

\bibitem{Hua2018}
M.~{Hua}, Y.~{Wang}, Z.~{Zhang}, C.~{Li}, Y.~{Huang}, and L.~{Yang}, ``Outage
  probability minimization for low-altitude {UAV}-enabled full-duplex mobile
  relaying systems,'' \emph{\textit{China Commun.}}, vol.~15, no.~5, pp. 9--24,
  May 2018.

\bibitem{Lu2015}
X.~Lu, P.~Wang, D.~Niyato, D.~I. Kim, and Z.~Han, ``Wireless networks with rf
  energy harvesting: A contemporary survey,'' \emph{\textit{{IEEE} Commun.
  Surveys Tuts.}}, vol.~17, no.~2, pp. 757--789, 2015.

\bibitem{PonnimbadugePerera2018}
T.~D. Ponnimbaduge~Perera, D.~N.~K. Jayakody, S.~K. Sharma, S.~Chatzinotas, and
  J.~Li, ``Simultaneous wireless information and power transfer ({SWIPT}):
  Recent advances and future challenges,'' \emph{\textit{{IEEE} Commun. Surveys
  Tuts.}}, vol.~20, no.~1, pp. 264--302, 2018.

\bibitem{Yang2018}
L.~Yang, J.~Chen, M.~O. Hasna, and H.-C. Yang, ``Outage performance of
  {UAV}-assisted relaying systems with rf energy harvesting,''
  \emph{\textit{{IEEE} Commun. Lett.}}, vol.~22, no.~12, pp. 2471--2474, 2018.

\bibitem{Wu2017}
Q.~Wu, G.~Y. Li, W.~Chen, D.~W.~K. Ng, and R.~Schober, ``An overview of
  sustainable green {5G} networks,'' \emph{\textit{{IEEE} Wireless Commun.}},
  vol.~24, no.~4, pp. 72--80, 2017.

\bibitem{Zou2016}
Y.~Zou, J.~Zhu, X.~Wang, and L.~Hanzo, ``A survey on wireless security:
  Technical challenges, recent advances, and future trends,''
  \emph{\textit{Proc. {IEEE}}}, vol. 104, no.~9, pp. 1727--1765, 2016.

\bibitem{Wu2019}
Q.~Wu, W.~Mei, and R.~Zhang, ``Safeguarding wireless network with {UAV}s: A
  physical layer security perspective,'' \emph{\textit{arXiv preprint
  arXiv:1902.02472}}, 2019.

\bibitem{Wang2017}
Q.~Wang, Z.~Chen, W.~Mei, and J.~Fang, ``Improving physical layer security
  using {UAV}-enabled mobile relaying,'' \emph{\textit{IEEE Wireless Commun.
  Lett.}}, vol.~6, no.~3, pp. 310--313, 2017.

\bibitem{Cai2019}
Y.~Cai, Z.~Wei, R.~Li, D.~W.~K. Ng, and J.~Yuan, ``Energy-efficient resource
  allocation for secure {UAV} communication systems,'' \emph{\textit{arXiv
  preprint arXiv:1901.09308}}, 2019.

\bibitem{Wang}
Y.~Wang, W.~Yang, X.~ShaYuan, and Y.~Cai, ``Energy-efficient secure
  transmission for {UAV}-enabled wireless powered communication,'' in
  \emph{\textit{Proc. 10th Int. Conf. Wireless Commun. Signal Process.
  (WCSP)}}, Oct 2018, pp. 1--5.

\bibitem{Sun2019}
X.~Sun, W.~Yang, Y.~Cai, Z.~Xiang, and X.~Tang, ``Secure transmissions in
  millimeter wave {SWIPT} uav-based relay networks,'' \emph{\textit{IEEE
  Wireless Commun. Lett.}}, pp. 1--1, 2019.

\bibitem{Amorim2017}
R.~Amorim, H.~Nguyen, P.~Mogensen, I.~Z. Kovacs, J.~Wigard, and T.~B. Sorensen,
  ``Radio channel modeling for {UAV} communication over cellular networks,''
  \emph{\textit{IEEE Wireless Commun. Lett.}}, vol.~6, no.~4, pp. 514--517,
  2017.

\bibitem{Khuwaja2018}
A.~A. Khuwaja, Y.~Chen, N.~Zhao, M.-S. Alouini, and P.~Dobbins, ``A survey of
  channel modeling for {UAV} communications,'' \emph{\textit{{IEEE} Commun.
  Surveys Tuts.}}, vol.~20, no.~4, pp. 2804--2821, 2018.

\bibitem{Azari2018}
M.~M. Azari, F.~Rosas, K.~Chen, and S.~Pollin, ``Ultra reliable {UAV}
  communication using altitude and cooperation diversity,''
  \emph{\textit{{IEEE} Trans. Commun.}}, vol.~66, no.~1, pp. 330--344, 2018.

\bibitem{Zhu2018}
Y.~Zhu, G.~Zheng, and M.~Fitch, ``Secrecy rate analysis of {UAV}-enabled
  mm{W}ave networks using {M}at\'ern {H}ardcore point processes,''
  \emph{\textit{IEEE J. Sel. Areas Commun.}}, vol.~36, no.~7, pp. 1397--1409,
  2018.

\bibitem{Mamaghani2019}
M.~T. Mamaghani and R.~Abbas, ``Security and reliability performance analysis
  for two-way wireless energy harvesting based untrusted relaying with
  cooperative jamming,'' \emph{\textit{IET Commun.}}, vol.~13, no.~4, pp.
  449--459, 2019.

\bibitem{Laneman2004}
J.~N. Laneman, D.~N.~C. Tse, and G.~W. Wornell, ``Cooperative diversity in
  wireless networks: Efficient protocols and outage behavior,''
  \emph{\textit{IEEE Trans. Inf. Theory}}, vol.~50, no.~12, pp. 3062--3080,
  2004.

\bibitem{Yao2019}
J.~Yao and J.~Xu, ``Secrecy transmission in large-scale {UAV}-enabled wireless
  networks,'' \emph{\textit{arXiv preprint arXiv:1902.00836}}, 2019.

\bibitem{Gradshteyn2014}
I.~S. Gradshteyn and I.~M. Ryzhik, \emph{\textit{Table of integrals, series,
  and products}}.\hskip 1em plus 0.5em minus 0.4em\relax Academic press, 2014.

\bibitem{Mamaghani2017}
M.~T. {Mamaghani}, A.~{Mohammadi}, P.~L. {Yeoh}, and A.~{Kuhestani}, ``Secure
  two-way communication via a wireless powered untrusted relay and friendly
  jammer,'' in \emph{\textit{Proc. IEEE Global Commun. Conf.}}, Dec 2017, pp.
  1--6.

\bibitem{Cao2016}
K.~{Cao} and X.~{Gao}, ``Solutions to generalized integrals involving the
  generalized marcum ${Q}$-function with application to energy detection,''
  \emph{\textit{{IEEE} Commun. Lett.}}, vol.~20, no.~9, pp. 1780--1783, Sep.
  2016.

\bibitem{convexBoyd}
S.~Boyd and L.~Vandenberghe, \emph{\textit{Convex optimization}}.\hskip 1em
  plus 0.5em minus 0.4em\relax Cambridge university press, 2004.

\bibitem{Mamaghani2018}
M.~Tatar~Mamaghani, A.~Kuhestani, and K.-K. Wong, ``Secure two-way transmission
  via wireless-powered untrusted relay and external jammer,''
  \emph{\textit{{IEEE} Trans. Veh. Technol.}}, vol.~67, no.~9, pp. 8451--8465,
  2018.

\bibitem{Bjornson2013}
E.~{Bjornson}, M.~{Matthaiou}, and M.~{Debbah}, ``A new look at dual-hop
  relaying: performance limits with hardware impairments,''
  \emph{\textit{{IEEE} Trans. Commun.}}, vol.~61, no.~11, pp. 4512--4525,
  November 2013.

\bibitem{Moser2008}
S.~M. {Moser}, ``Expectations of a noncentral {C}hi-square distribution with
  application to i.i.d {MIMO} {G}aussian fading,'' in \emph{\textit{Proc. IEEE
  ISIT}}, Dec 2008, pp. 1--6.

\end{thebibliography}

\end{document}